\documentclass [twoside] {article}

\usepackage{amssymb,amsfonts,amsmath,graphicx}
\usepackage{natbib}
\usepackage{setspace,color,booktabs,multirow}

\newtheorem{theorem}{Theorem}
\newtheorem{lemma}{Lemma}

\newtheorem{corollary}
{Corollary}

\newtheorem{definition}{Definition}

\newenvironment{proof}[1][Proof:]{\begin{trivlist}
\item[\hskip \labelsep {\bfseries #1}]}{\end{trivlist}}
\newcommand{\qed}{\nobreak \ifvmode \relax \else
      \ifdim\lastskip<1.5em \hskip-\lastskip
      \hskip1.5em plus0em minus0.5em \fi \nobreak
      $\Box$\fi}

\begin{document}

%\runningheads{A.P. Oron et al.}{Convergence of Dose-Finding Designs}

%\onehalfspacing
\title{Convergence of Nonparametric Long-Memory Dose-Finding Designs}
\author{Assaf P. Oron, David Azriel and Peter D. Hoff}

%\address{\affilnum{1}Department of Environmental and Occupational Health Sciences, University of Washington, 4225 Roosevelt Way NE, Seattle WA 98105, U.S.A., email: assaf@uw.edu
%\affilnum{2}\ Department of Statistics, Hebrew University, Mt. Scopus, 91905, Jerusalem, Israel, email: david.azriel@mail.huji.ac.il
%\affilnum{3}\ Department of Statistics, University of Washington, Seattle WA 98195, U.S.A., email: pdhoff@uw.edu}
%
%\corraddr{Department of Environmental and Occupational Health Sciences, University of Washington, 4225 Roosevelt Way NE, Seattle WA 98105, U.S.A.; \emph{assaf@uw.edu}}

%\markboth{A.P. Oron and P.D. Hoff}{Nonparametric Long-Memory Phase~I Designs}
\maketitle

\begin{abstract}
We examine nonparametric dose-finding designs that use toxicity estimates based on all available data at each dose allocation decision. We prove that one such design family, called here ``interval design'', converges almost surely to the maximum tolerated dose (MTD), if the MTD is the only dose level whose toxicity rate falls within the pre-specified interval around the desired target rate. Another nonparametric family, called ``point design'', has a positive probability of not converging. In a numerical sensitivity study, a diverse sample of dose-toxicity scenarios was randomly generated. On this sample, the ``interval design'' convergence conditions are met far more often than the conditions for one-parameter design convergence (the Shen-O'Quigley conditions), suggesting that the interval-design conditions are less restrictive. Implications of these theoretical and numerical results for small-sample behavior of the designs, and for future research, are discussed.

%\begin{keywords}
Keywords: Adaptive Designs, Continual Reassessment Method, Dose Finding, Phase I clinical trials, Cumulative Cohort Design, Up-and-Down

\end{abstract}

\section{Introduction}

Dose-finding designs attempt to identify the dose for which only a given fraction $p$\  of the population experiences some adverse (e.g., toxic) response. This dose is often called the experiment's ``target'', and can be symbolically denoted $F^{-1}(p)$ where $F(x)$ is an adverse-response-rate curve, monotonically increasing with the dose strength $x$.   In practice, it is more common to seek the dose closest to target from among a pre-specified fixed set of dose levels. This is known as the maximum tolerated dose (MTD). Dose-finding designs self-correct the dose allocation, according to hitherto observed outcomes, and thus belong to the family of sequential designs.

Some dose-finding designs, known as ``rule-based'' or ``memoryless'' \citep{OQuigleyZohar06}, are characterized by fixed dose-transition rules based on a limited subset of available observations (usually the most recent ones), and without any assumptions about the dose-toxicity curve $F$. A prominent example is the `3+3' protocol \citep{Carter73}, used for the vast majority of Phase~I cancer trials. We will refer to this class as ``short-memory'' designs.  Another, recently popular approach, is called ``model-based'' or ``designs with memory''. Such designs incorporate a model for $F$, and allocate doses via an estimation procedure based on all available observations. We will call these designs ``long-memory''. The overwhelming majority of novel dose-finding designs appearing in recent literature are long-memory, with Bayesian designs taking center stage \citep{OQuigleyEtAl90,BabbEtAl98}. In Bayesian designs a parametric model curve $G\left(x,\theta,\phi\right)$ substitutes for $F$, with $\theta$ denoting data-estimable parameters and $\phi$ fixed prior parameters. In the most common implementation, the next cohort is chosen according to where $G\left(x,\hat{\theta},\phi\right)$ crosses the horizontal line $y=p$ (Figure \ref{fig:demo}, left).

Designs that do not clearly belong to the short-memory or long-memory types, have also been suggested. These include two-stage \citep{Storer01,Potter02} and hybrid designs \citep[Section 5.3]{IvanovaEtAl03,Oron07}. Yet another intermediate approach suggests using all available data to estimate $F$, i.e., it is long-memory, but avoids parametric or Bayesian model specification \citep{LeungWang01,YuanChappell04,IvanovaEtAl07}. This {\it nonparametric long-memory} design type is the subject of our article.

We focus on convergence of these designs. The term ``convergence'' applied to dose-finding does not usually refer to convergence of our estimate of $F$; the point estimates of $F$ at the doses are guaranteed to converge almost surely to their true value in the limit of infinite sample size, as will be proven below in Section \ref{sec:prelim}. Rather, convergence in the dose-finding context refers to {\it allocation convergence:}  the convergence of the sequence of allocated doses to some stationary pattern. Short-memory designs belonging to the up-and-down family \citep{DixonMood48} generate Markov chains of doses, converging at a geometric rate to a stationary random walk whose dose-allocation distribution is centered close to target. The properties of up-and-down designs can be analyzed using standard Markov chain theory \citep{Derman57,DurhamFlournoy94,Gezmu96,GezmuFlournoy06,OronHoff09}. As to long-memory designs, proofs of allocation convergence are few and far between. In fact, nearly all of the novel long-memory designs -- and dozens of them have been put forth since 1990 -- lack a convergence proof.

To date, we are aware of the following published long-memory convergence proofs:
\begin{itemize}
\item Shen and O'Quigley\citet{ShenOQuigley96} proved that the one-parameter frequentist analogue to the CRM design converges almost surely (at a root-$n$ rate) to the MTD, a notion that will be defined in Section \ref{sec:prelim}. This result is widely perceived as a generic convergence-under-misspecification proof for CRM. However, Cheung and Chappell \citet{CheungChappell02} demonstrate that in fact the proof requires rather restrictive conditions. This will be explored in more detail in Section \ref{sec:numer}.
\item Zacks et al. \citet{ZacksEtAl98} present a similar result, under different and arguably even tighter restrictions on the form of $F$.
\item Ivanova et al. \citet{IvanovaEtAl03} prove that a hybrid 1950's design attributed to Narayana converges to a two-level random walk around the MTD; however, for reasons probably related to undesirable early-stage behavior \citep{Oron05}, this design has not been mentioned since then.
\end{itemize}
None of these proofs applies to the nonparametric long-memory designs we examine here.

\section{Preliminaries}\label{sec:prelim}

\medskip
\subsection{The Designs}

\begin{figure}
\begin{center}
\includegraphics[width=0.95\textwidth,height=0.3\textwidth, trim=10 10 10 10 ]{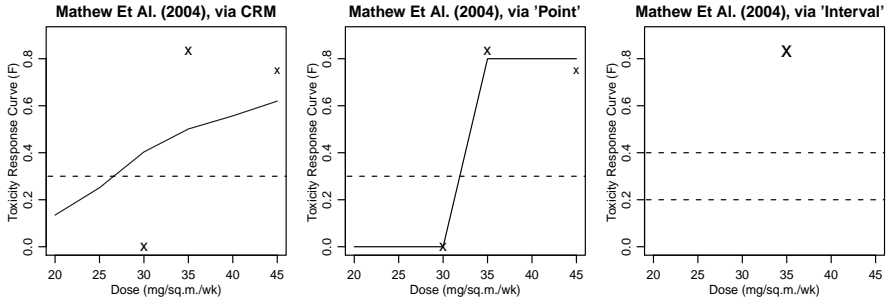}
\caption{Demonstration of the parametric CRM design (left), and the nonparametric point (center) and interval (right) designs, using data from the Mathew et al. 2004 experiment targeting $30\%$ toxicity \citep{MathewEtAl04}. Shown is the situation after Cohort 3, administered at 35~mg/m$^2$/week. Dose spacing was uniform at 5 mg/m$^2$/week. CRM (left) would allocate to Cohort 4 the dose closest to where the posterior model curve (solid line) crosses the dashed horizontal $y=0.3$ line -- i.e., 25~mg/m$^2$/week. The 'X' marks denote the actual observed toxicity rates. The point design (center) would follow the same principle, but using a nonparametric curve obtained via isotonic-regression interpolation; hence, Cohort 4 would receive 30 mg/m$^2$/week. The interval design (right) looks only at the actual toxicity rate at the dose Cohort 3 received ('X' mark). Since it falls above the interval marked by two dashed horizontal lines, the allocation will de-escalate one level to 30~mg/m$^2$/week. }\label{fig:demo}
\end{center}
\end{figure}

We describe the dose-finding problem via a latent-variable model: let $Y(x)\sim \mathrm{Bernoulli}\left(F(x)\right)$ be a binary toxicity response of some dose strength $x$, with the toxicity-rate function $F(x)$ strictly monotone increasing but not directly observable. The overall goal is to estimate the target $Q_p\equiv F^{-1}(p)$, which can be seen as the $100p$-th percentile of $F$ if one thinks of $F$ as a cumulative distribution function of toxicity thresholds. Consider a sequential design treating $k_j\geq 1$ subjects at cohort $j$, $j=1,2,\ldots$, with the value of the allocated-dose r.v. $X_j$ taken out of a set of $m$ predetermined dose levels $\Omega\equiv \left\{d_1,d_2,\ldots d_m : d_1<d_2<\ldots<d_m\right\}$. For simplicity and without loss of generality with respect to our proofs, from here on we assume that all cohorts are of size $1$, and index successive treatments as $X_i,i=1,\ldots n,\ldots$ The toxicity responses $Y_i$ are assumed independent given the $X_i$.

As mentioned in the introduction, rather than precisely estimate $Q_p$ researchers are often content with identifying the MTD, i.e., the dose level closest to $Q_p$ according to some distance criterion; we will denote the MTD as $d_{u^*}\in\Omega$. In this article we assume that distance on the response scale is used to find the MTD; in other words, $u^*\equiv\textrm{argmin}_{1\leq u\leq m}\left|F\left(d_u\right)-p\right|$.

All long-memory designs use the raw toxicity frequencies, which can be seen as Binomial point estimates of $F$ given $n$ observations,
\begin{equation}\label{eq:Fhat}
\hat{F}_n\left(d_u\right)\equiv\frac{\sum_{i=1}^n y_i\textbf{1}\left[X_i=d_u\right]}{\sum_{i=1}^n\textbf{1}\left[X_i=d_u\right]};\ \ u=1,\ldots m,
\end{equation}
where $y_i$ is the binary toxicity outcome ($0$ or $1$). While parametric designs use the $\hat{F}$ values indirecty as inputs to the calculation of $\hat{\theta}$, nonparametric long-memory designs use them directly, with a possible modification to ensure monotonicity of the $\hat{F}$ values using standard methods \citep{RobertsonEtAl88}.

Following are the definitions of the ``point'' and ``interval'' nonparametric designs.

\begin{definition}\label{def:designs}
(i) A ``point-based nonparametric long-memory'' Phase~I design (hereafter, ``point design'') starts at an arbitrary dose. At each subsequent step, the design allocates the next cohort to the level whose (possibly monotonized) $\hat{F}\left(d_u\right)$ value is closest to $p$. If the highest dose for which such an estimate is available maintains $\hat{F}\left(d_u\right)<p$, the experiment escalates to $d_{u+1}$ -- boundary permitting -- and vice versa.

(ii) An ``interval-based nonparametric long-memory'' Phase~I design (hereafter, ``interval design'') starts at an arbitrary dose. At each subsequent step, the design compares a (possibly monotonized) $\hat{F}\left(d_u\right)$, with $d_u$ being the currently-administered dose, to the interval $\left(p-\Delta p_1,p+\Delta p_2\right)$ with $\Delta p_1>0,\ \ \Delta p_2>0$ predetermined constants. If $$\hat{F}\left(d_u\right)\in\left(p-\Delta p_1,p+\Delta p_2\right),$$ $d_u$ will be administered again. If $\hat{F}\left(d_u\right)\leq p-\Delta p_1,$ $d_{u+1}$ will be administered (unless $u=m$ in which case $d_m$ will be administered again), and vice versa if $\hat{F}\left(d_u\right)\geq p+\Delta p_2$.

For both design types, the recommended MTD is the next dose level that would have been allocated at the experiment's end, had another cohort been run.
\end{definition}

The point design was suggested by Leung and Wang \citet{LeungWang01}; it is a direct variation on parametric designs such as CRM, with the parametric curve $F(\hat{\theta})$ replaced by a monotone nonparametric interpolation of $F$ between dose levels (Figure \ref{fig:demo}, center). The interval design's principle is different; one might call it ``narrow long-memory'' since the allocation decision is based on prior outcomes at the current dose only (Figure \ref{fig:demo}, right). Rather than look for some optimal dose at each cohort, the allocation would repeat the existing dose as long as the toxicity frequency at that dose falls within the interval. Different versions of the interval design were put forth by Yuan and Chappell \citet{YuanChappell04} and Ivanova et al \citet{IvanovaEtAl07}. The interval design does not allow for skipping dose levels between consecutive cohorts.

\subsection{Allocation Convergence}

We now clarify the meaning of allocation convergence using the terminology introduced above. The sample space for allocation convergence is the space of all permissible infinite sequences of assigned doses, which is a subset of $\Omega^\infty$ (usually subject to the the constraint of no dose skipping). Each design induces a probability distribution on sequences in this sample space (probabilities of finite subsequences can be exactly calculated, with knowledge of the design's rules and of $F$ values at the doses). Almost sure convergence to the MTD means that sequences ending with infinite and uninterrupted repetitions of $d_{u^*}$ have a combined probability of $1$.

On this sample space, define the random set
\begin{equation}\label{eq:S}
\mathbb{S}\equiv\left\{u:n_u\to\infty\ \textrm{ as }n\to\infty\right\},
\end{equation}

where $n_u$ is the number of subjects assigned to $d_u$. In words, $\mathbb{S}$ is the set of indices for levels appearing an infinite number of times in the sequence. Obviously $\mathbb{S}$ is nonempty for all sequences in the sample space, all being infinite. Moreover, since the interval design does not allow for dose skipping $\mathbb{S}$ must be connected, i.e., composed of consecutive levels. Thus, the value of $\mathbb{S}$ for different sequences in the sample space can be described via an ordered pair of integer random variables $S_1\leq S_2:\ \mathbb{S}=S_1,\ldots S_2$.

%%%%%%%%%%%%%%% THE POINT-ESTIMATE-CONVERGENCE LEMMA (AZRIEL)

We end the preliminaries with the following point-estimate convergence result, which holds regardless of the type of design.

\begin{lemma} \label{lem:SLLN}
${\hat F}_n(d_u) \rightarrow F(d_u)$ as $n\to\infty$, almost surely for all $u\in\mathbb{S}$ .
\end{lemma}
\begin{proof} First, by definition of $\mathbb{S}$ we know that for all $u\in\mathbb{S}$, $n_u\to\infty$ as $n\to\infty$. Second, note that the point estimates can be written as

\begin{equation}\label{eq:fhatf}
{\hat F}_n(d_u) =  F(d_u) + \frac{1}{n_u} \sum_{i=1}^n I(X_i=d_u)(Y_i - F(d_u)).
\end{equation}

Now,   $M_n \equiv \sum_{i=1}^n I(X_i=d_u)(Y_i -F(d_u))$ is a square integrable martingale with respect to the filtration ${\cal F}_n \equiv \sigma(X_1,Y_1,\ldots,X_{n},Y_{n})$. Its quadratic variation is:
\[
\sum_{i=1}^n [I(X_i=d_u)]^2 \cdot F(d_u)\cdot [1-F(d_u)] \propto \sum_{i=1}^n I(X_i=d_u)=n_u.
\]
Therefore, due to the strong law of martingales (ref., \citep{Shirayev96}, p. 519, theorem 4):
\begin{equation*}\label{eq:shirayev}
\frac{1}{n_u} \sum_{i=1}^n I(X_i=d_u)(Y_i -F(d_u)) \rightarrow 0~~a.s.~~\forall u\in\mathbb{S}.
\end{equation*}
Revisiting (\ref{eq:fhatf}), the lemma's statement immediately follows.
\qed\end{proof}

%%%%%%%%%%%%%%%%%%%%%%%%%%%

\section{Interval-Design Convergence}\label{sec:interval}

Note that with respect to allocation convergence, the space of possible configurations of $\mathbb{S}$ can be partitioned into three major subspaces or events $A,B,C$:

%\Alph{enumi}
\begin{itemize}
\item $A:\ s_1=s_2=u^*$,
\item $B:\ s_1<s_2$ and $u^*\in \mathbb{S}$,
\item $C:\ u^*\not\in \mathbb{S}$.
\end{itemize}
Almost sure convergence is equivalent to stating that $\Pr(A)=1$.

%%%%%%%%%% MAIN THEOREM

\begin{theorem}\label{theorem:interval}
(i) Dose allocations in interval designs converge almost surely to $d_{u^*}$, if the latter maintains
\begin{equation}\label{eq:interval}
F\left(d_{u^*}\right)\in\left(p-\Delta p_1,p+\Delta p_2\right),
\end{equation}
and if $d_{u^*}$ is also the \underline{only} level satisfying
\begin{equation}\label{eq:interval2}
F\left(d_{u^*}\right)\in\left[p-\Delta p_1,p+\Delta p_2\right].
\end{equation}

(ii) Almost-sure convergence to $d_{u^*}$ will also occur if $F\left(d_1\right)\geq p+\Delta p_2$ (meaning that $u^*=1$) or $F\left(d_m\right)\leq p-\Delta p_1$ (meaning that $u^*=m$).
\end{theorem}

\begin{proof} (i) We begin by showing that $\Pr(C)=0$, which is equivalent to $\Pr\left(u^*\in \mathbb{S}\right)=1$. We do it by contradiction, assuming w.l.o.g. that there is some specific level $s_1>u^*$ for which $\Pr\left(S_1=s_1>u^*\right)>0$ (in other words, that there are sequences with a positive probability of occurring, in which beyond a certain point only levels above the MTD are visited). From the theorem's assumptions, we know that $F\left(d_{s_1}\right)>p+\Delta p_2$. Due to Lemma \ref{lem:SLLN}, this means that for $n$ large enough and all sequences described by the conditioned event,
\begin{equation}\label{eq:false}
\Pr\left\{\hat{F}_n\left(d_{s_1}\right)>p+\Delta p_2\mid S_1=s_1>u^*\right\}=1.
\end{equation}
Given the interval design's transition rules, this means that eventually the next-lower level, $d_{s_1-1}$, will be allocated following each visit to $d_{s_1}$ with probability $1$ conditioned on the above event.\footnote{Any monotonizing modifications to the $\hat{F}$'s do not matter as $n\to\infty$, since in that limit they are needed with probability zero. Or, if they involve a level not belonging to $\mathbb{S}$, their impact tends to zero.} It follows that $\left(s_1-1\right)\in\mathbb{S}$, reaching a contradiction. We conclude that there is no $s_1>u^*$ for which $\Pr\left(S_1=s_1>u^*\right)>0$, and therefore one cannot condition on such an event as was done in (\ref{eq:false}) -- and similarly, no $s_2<u^*$ for which $\Pr\left(S_2=s_2<u^*\right)>0$. In terms of the partition of sequence space defined above, $\Pr\left(u^*\in \mathbb{S}\right)=\Pr\left(A\cup B\right)=1$.

Now we can assume that $u^*\in \mathbb{S}$. Given the theorem's conditions and according to Lemma \ref{lem:SLLN}, eventually for $n$ large enough $$\Pr\left(\hat{F}_n\left(d_{u^*}\right)\in\left(p-\Delta p_1,p+\Delta p_2\right)\mid u^*\in S\right)=1,$$
and so upon the next visit to $d_{u^*}$ it will be repeatedly allocated with probability $1$. This means, that with probability one (conditional upon $u^*\in \mathbb{S}$) there can be no other level in $\mathbb{S}$. Therefore $\Pr(A)=1$, and the interval design converges almost surely.

(ii) In the same vein, if all true $F$ values of design levels are above or below the target interval, then with probability $1$ the boundary level with $F$ value closest to the interval ($d_1$ or $d_m$) belongs to $\mathbb{S}$, and eventually the design will repeatedly hit upon the boundary condition mandating repetition of that level with probability $1$ as well.\qed
\end{proof}

%%%%%%%%%% PROOF END; NOW SOME COROLLARIES

The theorem's proof itself suggests what might happen in case its conditions are violated. Hence, the following two results are immediate.

\begin{corollary}\label{col:1} (i) If no dose level satisfies (\ref{eq:interval}) but $p\in\left[F\left(d_1\right),F\left(d_m\right)\right]$, an interval design would eventually oscillate with probability $1$ between the two doses whose $F$ values straddle the target interval.

(ii) If there is more than one level satisfying (\ref{eq:interval}), with probability $1$ an interval design will converge to either of these levels. However, convergence to $d_u^*$ itself is not guaranteed.

\end{corollary}

\section{Point-Design Convergence}\label{sec:point}
The point design has a positive probability of not converging. We show this via a simple, yet generic counterexample: Assume that $p<1/2$, $F\left(d_{u^*}\right)=p$ and $F\left(d_{u^*-1}\right)<p$ (all other levels matter little). The experiment uses cohorts of size $k\geq 1$ and proceeds, as dose-finding trials often do, from below. Sooner or later $d_{u^*-1}$ is reached, and with high probability within a few cohorts we will have $\hat{F}\left(d_{u^*-1}\right)<p$, mandating escalation to $d_{u^*}$. Now, suppose that the very first cohort at $d_{u^*}$ is all toxicities; clearly the probability for that occurring is positive. Then $\hat{F}\left(d_{u^*}\right)=1$. Since $p<1/2$, regardless of the value of $\hat{F}\left(d_{u^*-1}\right)$, it is now closer to target than $\hat{F}\left(d_{u^*}\right)$ and it will be assigned. Moreover, since $\hat{F}\left(d_{u^*}\right)=1$ monotonicity at $d_{u^*}$ will not be violated, meaning that no monotonizing corrections can modify this point estimate, which will remain too far from $p$ for the remainder of the experiment. Hence $d_{u^*}$ will never be assigned again. A similar argument was made by Cheung in a Biometrics letter to the editor \citep{Cheung02}.

\section{Numerical Sensitivity Study}\label{sec:numer}
\subsection{Convergence of One-Parameter Designs}

How restrictive are the conditions outlined in Theorem \ref{theorem:interval}? One way to gauge this is to compare them with the conditions of Shen and O'Quigley's proof for CRM-like one-parameter frequentists designs \citep{ShenOQuigley96}. We now revisit its conditions in some detail. Beside straightforward conditions guaranteeing that the modeled dose-toxicity curve $G\left(x,\theta\right)$ can match the true curve $F$ at least at one $x$ value by changing the value of $\theta$, the proof focuses on how well $G$ fits  $F$ elsewhere. Being a one-parameter model, $G$ cannot be guaranteed to do so simultaneously at more than one point. Moreover, the choice of this point uniquely determines $\theta$. Suppose w.l.o.g. that $G\left(d_u\right)=F\left(d_u\right)$, and call the resulting parameter value $\theta_u$. Then the level which, according to $G\left(\theta_u\right)$ appears to be the MTD, will be called the level \emph{``nominated''} by $d_u$ (since it will be allocated whenever $G$ matches $F$ at $d_u$). The crucial and most restrictive Shen-O'Quigley condition is {\it that all levels nominate the true MTD.} Hereafter we will refer to this convergence result as ``CRM convergence'', even though it is in fact a proof for the convergence of an analogous frequentist design.

\smallskip

Cheung and Chappell, in their interpretation of the proof, opine that this requires a very close match between $G$ and $F$ along the entire dose range \citep{CheungChappell02}. They go on to suggest that perhaps the Shen-O'Quigley conditions were too restrictive, and it might be enough for the MTD to nominate itself, and additionally doses below the MTD nominate higher doses than themselves, and vice versa. Thus, dose allocation might eventually be ``funnelled'' towards the MTD. The conjecture has not been proven. Conversely, it is clear when a one-parameter Bayesian design {\it cannot} be guaranteed to converge:

\begin{enumerate}
\item When the MTD fails to nominate itself; or
\item When other levels beside the MTD nominate themselves; or
\item When the ``funneling'' conditions suggested by Cheung and Chappell are not met.
\end{enumerate}

\subsection{Comparing Interval-Design and CRM Convergence}

We explored numerically the relative restricted-ness of the two sets of conditions. Since the convergence of both the interval design and CRM can be directly determined from the values of $F$ at the dose levels, together with the interval endpoints and the parameters of $G$, there is no need to simulate actual experiments. Rather, we simulated various scenarios of $F$ on $m=5$ and $m=10$ dose sets, and examined whether the CRM and the interval design convergence conditions are met for each scenario. We chose the target $p=0.3$, the value most commonly used in Phase I cancer trials, which is the application for which both designs have been developed. For this target, developers of the CCD interval design recommend the interval $(0.2,0.4)$ when $m=6$; they provide no recommendation for other values of $m$. We also explored the narrower interval $(0.25,0.35)$; hereafter we refer to the interval design in this simulation as ``CCD''. For CRM, we used the recently popular ``power'' model, in which $G\left(d_1,\ldots,d_m;\theta\right)=\left(p_1,\ldots,p_m\right)^{\textrm{exp}(\theta)}$, with the $p$'s being prior toxicity rates assigned to each dose. Experienced CRM designers do not choose these rates solely according to toxicity information knowledge, but mostly in order to ensure sensible small-sample behavior. A choice commonly encountered in the field resembles a geometrically-increasing sequence, e.g. $\left(p_1,\ldots,p_m\right)=\left\{0.05,0.1,0.2,0.4,0.8\right\}$ for $m=5$ \citep{PistersEtAl04}.

In order to generate reasonably realistic scenarios without restricting ourselves to a given distribution family, and also in order to minimize the direct impact of arbitrary conscious choice upon $F$, we simulated increments of $F$ in each scenario as a random Dirichlet vector. Dirichlet distribution parameters control the likelihood of generating various curves; these parameters themselves were randomly drawn out of a finite pool, producing a range of diverse, yet reasonably realistic $F$ curves, which would be relevant for the dose-finding problem as defined here. Additionally, lower and upper bounds were placed on increments of $F$ to exclude scenarios in which adjacent-dose toxicity rates are virtually indistinguishable, or spaced too far apart. Figure \ref{fig:simul} shows a random sample of $20$ scenarios (out of $2500$ used in the study) for each of $m=5$ and $m=10$. Scenarios were simulated and convergence evaluated using the R language version 2.9.1 \citep{RLang}. Additional details appear in the supplementary material.

\begin{figure}
\begin{center}
\includegraphics[width=0.9\textwidth,height=0.45\textwidth, trim=10 10 10 10 ]{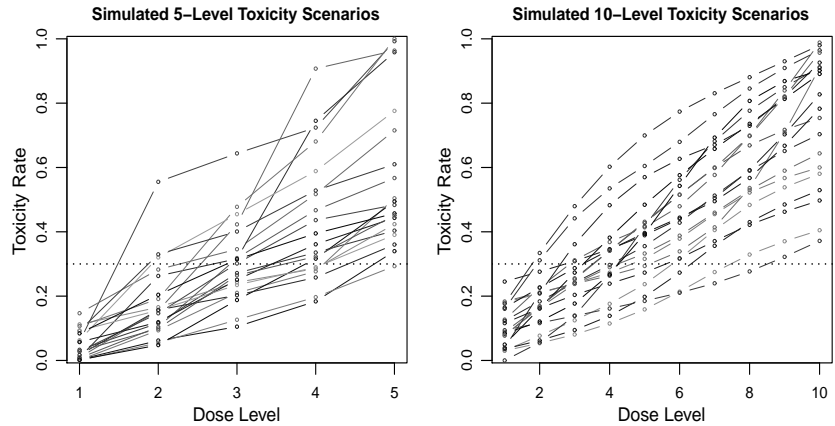}
\caption{Simluated dose-toxicity curves. Random samples of 20 out of the 2500 simulated scenarios, for $m=5$ (left) and $m=10$ (right).}\label{fig:simul}
\end{center}
\end{figure}

For CCD, we distinguish between the three possible convergence outcomes proven in Theorem 1 and Corollary \ref{col:1}:
\begin{itemize}
\item Convergence to the MTD guaranteed. The MTD is the only level in the interval, or the target is below/above the design dose range (column marked ``Yes'').
\item More than one level in  the interval, and hence only convergence to within the interval is assured, but not to the MTD itself (column marked ``No: 2+'').
\item No level in the interval, and hence an asymptotic oscillating behavior is expected (column marked ``No: 0'').
\end{itemize}

For CRM, we distinguish between five possible outcomes:
\begin{itemize}
\item Convergence to the MTD guaranteed: all levels nominate the MTD (row marked ``Yes'').
\item ``Soft convergence'': convergence not guaranteed by proof, but the Cheung-Chappell ``funneling'' conditions are met (row marked ``Funneling'').
\item Convergence not guaranteed: one of the three failure modes outlined earlier (rows marked ``No: 0'', ``No: 2+'' and ``No Funneling'').
\end{itemize}

Below are tables of simulation results for five and ten design levels. The full Shen-O'Quigley conditions for CRM convergence are only rarely met. The weaker ``funneling'' conditions are met in quarter of the $m=5$ cases, but nearly half of the $m=10$ cases. A notable observation is that only two of the three CRM failure modes take place, at least in these simulations; hence, only four outcomes are tabulated instead of five. The missing entry is ``No Funneling'': if funneling is violated, then always (in our simulation runs) one of the other two conditions is violated as well.

\begin{table}
\caption{Comparative theoretical convergence summary of CRM and CCD designs, for a diverse ensemble of numerically-generated scenarios. The MTD was defined as the dose level whose true $F$ value is closest to $0.3$. All numbers in the table are in percents. Row and column labels are explained in the text.}

\smallskip

\begin{tabular}{cp{1.5cm}crrr|rrr}

  \toprule
\large{\bf $m=5$} & & &\multicolumn{3}{c}{{\bf CCD (width: $\pm 0.1$)}}&\multicolumn{3}{|c}{{\bf CCD (width: $\pm 0.05$)} } \\
\midrule
& & & No: 0 & No: 2+ & Yes & No: 0 & No: 2+ & Yes\\
  \cmidrule{4-9}

 \multicolumn{2}{l}{\bf CCD Margins} & & {\bf  7.6} & {\bf 56.2 } & {\bf 36.2} & {\bf 32.6 } & {\bf  13.4  } & {\bf 53.9 } \\
\midrule
\multirow{5}{*} {{\bf CRM}} & & {\bf CRM Margins} &&&&& \\
\cmidrule{3-3}
& No: 0 & {\bf 17.0} & 7.0 & 2.1 & 7.8 & 16.6 & 0.0 & 0.4 \\
& No: 2+ & {\bf 56.3}     & 0.3 & 46.4 & 9.6 & 9.4  & 13.4 & 33.6 \\
& ``Funneling'' & {\bf  25.0} &  0.1 & 7.5 & 17.4 & 6.0 & 0.1 &18.8 \\
& Yes & {\bf 1.7 } & 0.2 & 0.1 & 1.4 & 0.6 & 0.0 & 1.1 \\
\bottomrule
\large{\bf $m=10$} & & &\multicolumn{3}{c}{{\bf CCD (width: $\pm 0.1$)}}&\multicolumn{3}{|c}{{\bf CCD (width: $\pm 0.05$)} } \\
\midrule
& & & No: 0 & No: 2+ & Yes & No: 0 & No: 2+ & Yes\\
  \cmidrule{4-9}

 \multicolumn{2}{l}{\bf CCD Margins} & & {\bf  0.8} & {\bf 92.2 } & {\bf 7.0} & {\bf 4.6 } & {\bf  53.6  } & {\bf 41.8 } \\
\midrule
\multirow{5}{*} {{\bf CRM}} & & {\bf CRM Margins} &&&&& \\
\cmidrule{3-3}
& No: 0 & {\bf 13.4} & 0.0 & 11.5 & 1.9 &  3.1 & 2.9 & 7.3 \\
& No: 2+ & {\bf 41.4}     & 0.0 & 41.3 & 0.2 & 0.1 & 36.6 & 4.8 \\
& ``Funneling'' & {\bf  44.4} &   0.4 & 39.4 & 4.7 & 0.9 & 14.1 & 29.4 \\
& Yes & {\bf 0.8 } & 0.4 & 0.1 & 0.3 & 0.5 & 0.0 & 0.3 \\
\bottomrule\end{tabular}
\end{table}

Observing the CCD results, exact convergence to the MTD is guaranteed in a far larger number of cases than with CRM. Together with the multiple-level (``No: 2+'') cases, in the vast majority of simulated scenarios CCD is guaranteed to converge to within the specified interval. Comparing the narrower and wider interval design options, we see that the former performs better with more design levels. This suggests that the density of dose levels should also play a role in determining the interval width, a point somewhat side-stepped by CCD's developers.

Comparing the two design approaches by case, there is a strong association between the CRM's failure to converge due to multiple self-nominating levels, and CCD's failure to converge due to multiple levels inside the interval: practically all scenarios indicating the former, also indicate the latter (but not vice versa). The ``funneling'' scenarios are associated with CCD scenarios that converge to within the interval.

\section{Discussion}\label{sec:conc}

\subsection{Convergence and the Role of Simulation}

As mentioned in the introduction, the explosion in novel long-memory design development lacks an accompanying effort to prove design convergence. However, as pointed out in the introduction, the common thread between all dose-finding designs is a self-correction mechanism to concentrate treatments around target. If this mechanism is sound, then it should eventually converge to some stationary behavior with desirable properties vis-a-vis the MTD. If convergence cannot be guaranteed under realistic conditions, then the self-correction mechanism itself is suspect regardless of sample size. In other words, convergence should be viewed as \emph{a necessary condition for dose-finding designs} (albeit not a sufficient one, since small-sample behavior does need to be examined separately). Therefore, the study of convergence should play an larger role in the field of novel dose-finding designs.

As to the use of simulation itself, we attempted to try and minimize the effect of direct human choice on the tested scenarios. Ivanova and coworkers \citep{IvanovaEtAl07} started along this direction, choosing $F$ values out of an ordered uniform distribution. We believe our approach further expands the horizons for a distribution-free simulation study, and does succeed in sampling a sizable region of the space of distributions that would be considered realistic by researchers in the field. It might serve as an initial template for future benchmark comparative performance simulations between designs, of the type that is common in fields such as machine learning.

\subsection{Implications for Interval Designs and CRM}

The results summarized in Table 1 underscore Cheung and Chappell's observation that the existing Shen-O'Quigley CRM convergence proof requires rather restrictive conditions. CRM convergence occurs only in a small fraction of simulated cases, and we venture to suggest the conditions for it are only rarely met in practice. Thus, an accurate interpretation of the Shen-O'Quigley result seems to be that CRM's convergence under \emph{correct} specification -- which for one-parameter models is an immediate result of standard MLE convergence theorems -- can be extended to very mildly misspecified models.\footnote{One way to quantify the degree of mis-specification is by measuring the total variation distance between $F$ and $G$, with $\theta$ chosen such that $F=G$ at the MTD. With $m=5$ the $G$ curves satisfying the Shen-O'Quigley conditions are about half as far, on the average, from $F$ as the other curves. With $m=10$ the difference is approximately threefold on the average.} The Cheung-Chappell ``funnelling'' conditions seem far more realistic than the Shen-O'Quigley conditions; perhaps this observation will serve as motivation for finding a proof that they indeed guarantee convergence. At this moment it is unclear whether such a proof is feasible, or what additional conditions it would require.

The CCD interval design converged to the MTD fairly often. The prospects of convergence improve with a well-informed choice of interval width. Our theoretical and numerical results suggest a far simpler approach to optimal interval-width choice than the intensive multi-scenario small-sample numerical study by the method's originators \citep{IvanovaEtAl07}. Based on Theorem \ref{theorem:interval}, study designers should aim to capture exactly one dose level in the interval, and erring towards more than one level is probably more desirable than capturing none. In the absence of prior scientific knowledge about the slope of $F$ around target, a total interval width of $1/m$ would do, as long as researchers believe that all $m$ levels have toxicity rates not too close to either $0$ or $1$. In any case, even when CCD does not converge to the MTD -- whether due to multiple levels in the interval, or none -- one can still guarantee a predictable asymptotic behavior with respect to the pre-specified interval. This is not the case with parametric designs in general and with one-parameter CRM in particular.

\subsection{Convergence, Small-Sample Behavior and Simulation}

Convergence studies can also shed light upon designs' small-sample behavior. For example, the up-and-down designs mentioned in the introduction converge at a geometric rate: their short memory facilitates a very quick self-correction mechanism. However, the self-correction is blunt: asymptotic behavior meanders around target, typically spreading most allocations over $2-4$ levels.\footnote{In spite of the blunt allocation distribution, up-and-down {\it estimates} do become sharper with time, since they rely on all the gathered information.} Long-memory designs converge (if they converge) much more slowly, at a root-$n$ rate; this means that their self-correction, even at small samples, is also slow. The promised compensation is a perfectly sharp asymptotic allocation distribution, zooming in on the MTD itself. Regardless of design and proof details, this outcome hinges upon the precision of point estimates, whose convergence was demonstrated in Lemma \ref{lem:SLLN}. Unfortunately, simple arithmetic on Binomial probabilities suggests that for point estimates to be precise with high reliability requires many more trials than the typical dose-finding sample size of $10-40$ subjects, who are inevitably spread over several dose levels.

In fact, the main difference between asymptotic and small-sample behavior, is that the latter is dominated by very imprecise Binomial point estimates. Thus, during an initial stage of the experiment, a long-memory design might point towards a level quite far from the MTD, if a large enough proportion of the individual-subject toxicity trials yielded ``atypical'' outcomes. In sampling terms, if the initial group of sampled toxicity thresholds can be seen collectively as an outlier, then long-term designs (both parametric and nonparametric) are led astray. When this happens, the long memory and its associated slow self-correction become liabilities rather than assets: point estimates are off, and they will now improve only gradually because the initial, ``outlier'' group of outcomes is still included in any subsequent estimate. Meanwhile, the design will insist upon collecting information at the wrong place, further slowing the self-correction mechanism. In practice, this positive-feedback reaction makes long-memory designs less robust to the experiment's first few observations, compared with short-memory designs. This phenomenon is unrelated to design details (parametric or nonparametric), and has been observed numerically for both CCD and CRM \citep{Oron07jsm,Oron09dae}. It underscores the two messages conveyed here: 1. The study of convergence properties can help explain small-sample behavior, and 2. Convergence is a necessary requirement for a sound dose-finding design, but convergence alone is not sufficient to guarantee desirable small-sample behavior.

\bibliographystyle{elsart-harv}
\bibliography{phdplus}

\begin{thebibliography}{29}
\expandafter\ifx\csname natexlab\endcsname\relax\def\natexlab#1{#1}\fi
\expandafter\ifx\csname url\endcsname\relax
  \def\url#1{\texttt{#1}}\fi
\expandafter\ifx\csname urlprefix\endcsname\relax\def\urlprefix{URL }\fi

\bibitem[{Babb et~al.(1998)Babb, Rogatko, and Zacks}]{BabbEtAl98}
Babb, J., Rogatko, A., Zacks, S., 1998. Cancer phase {I} clinical trials:
  Efficient dose escalation with overdose control. Stat. Med. 17, 1103--1120.

\bibitem[{Carter(1973)}]{Carter73}
Carter, S., 1973. Study design principles in the clinical evaluation of new
  drugs as developed by the chemotherapy programme of the {N}ational {C}ancer
  {I}nstitute. In: The Design of Clinical Trials in Cancer Therapy. Editions
  Scientific Europe, Brussels, pp. 242--289.

\bibitem[{Cheung(2002)}]{Cheung02}
Cheung, Y., 2002. Letter to the editor. Biometrics 58~(1), 237--240.

\bibitem[{Cheung and Chappell(2002)}]{CheungChappell02}
Cheung, Y., Chappell, R., 2002. A simple technique to evaluate model
  sensitivity in the continual reassessment method. Biometrics 58~(3),
  671--674.

\bibitem[{Derman(1957)}]{Derman57}
Derman, C., 1957. Non-parametric up-and-down experimentation. Ann. Math.
  Statist. 28, 795--798.

\bibitem[{Dixon and Mood(1948)}]{DixonMood48}
Dixon, W.~J., Mood, A.~M., 1948. A method for obtaining and analyzing
  sensitivity data. J. Am. Statist. Assoc. 43, 109--126.

\bibitem[{Durham and Flournoy(1994)}]{DurhamFlournoy94}
Durham, S.~D., Flournoy, N., 1994. Random walks for quantile estimation. In:
  Statistical decision theory and related topics V (West Lafayette, IN, 1992).
  Springer, New York, pp. 467--476, qA279.4 .S745 1994.

\bibitem[{Gezmu(1996)}]{Gezmu96}
Gezmu, M., 1996. The geometric up-and-down design for allocating dosage levels.
  Ph.D. thesis, American University, Washington, DC, USA.

\bibitem[{Gezmu and Flournoy(2006)}]{GezmuFlournoy06}
Gezmu, M., Flournoy, N., 2006. Group up-and-down designs for dose-finding. J.
  Statist. Plann. Inference 136~(6), 1749--1764.

\bibitem[{Ivanova et~al.(2007)Ivanova, Flournoy, and Chung}]{IvanovaEtAl07}
Ivanova, A., Flournoy, N., Chung, Y., 2007. Cumulative cohort design for
  dose-finding. J. Statist. Plann. Inference 137, 2316--2327.

\bibitem[{Ivanova et~al.(2003)Ivanova, Haghighi, Mohanty, and
  Durham}]{IvanovaEtAl03}
Ivanova, A., Haghighi, A., Mohanty, S., Durham, S., 2003. Improved up-and-down
  designs for phase {I} trials. Stat. Med. 22, 69--82.

\bibitem[{Leung and Wang(2001)}]{LeungWang01}
Leung, D. H.-Y., Wang, Y.-G., 2001. Isotonic designs for phase {I} trials.
  Cont. Clin. Trials 22~(2), 126--138.

\bibitem[{Mathew et~al.(2004)}]{MathewEtAl04}
Mathew, P., et~al., 2004. Platelet-derived growth factor receptor inhibitor
  imatinib mesylate and docetaxel: {A} modular phase {I} trial in
  androgen-independent prostate cancer. J Clin Oncol 22~(16), 3323--3329.

\bibitem[{O'Quigley et~al.(1990)O'Quigley, Pepe, and Fisher}]{OQuigleyEtAl90}
O'Quigley, J., Pepe, M., Fisher, L., 1990. Continual reassessment method: a
  practical design for {P}hase {I} clinical trials in cancer. Biometrics
  46~(1), 33--48.

\bibitem[{O'Quigley and Zohar(2006)}]{OQuigleyZohar06}
O'Quigley, J., Zohar, S., 2006. Experimental designs for phase {I} and phase
  {I/II} dose-finding studies. Brit. J. Canc. 94~(5), 609--613.

\bibitem[{Oron(2005)}]{Oron05}
Oron, A.~P., 2005. The up-and-down experimental method: Stochastic properties
  and estimators. Dependent data preliminary exam report, Statistics
  Department, University of Washington, Seattle, WA, USA, revised version,
  March 2005; available from the author.

\bibitem[{Oron(2007)}]{Oron07}
Oron, A.~P., 2007. Up-and-down and the percentile-finding problem. Ph.D.
  thesis, University of Washington, Seattle, Washington, USA, available on
  http://arxiv.org/abs/0808.3004.
\newline\urlprefix\url{http://arxiv.org/abs/0808.3004}

\bibitem[{Oron(2009)}]{Oron09dae}
Oron, A.~P., 2009. Small-sample behavior of long-memory {P}hase {I} cancer
  designs. In: DAE 2009, Columbia, MO. Invited Conference Talk.

\bibitem[{Oron and Hoff(2007)}]{Oron07jsm}
Oron, A.~P., Hoff, P.~D., 2007. Bayesian up-and-down: combining the best of
  both worlds. In: Joint Statistical Meetings, Salt Lake City. Conference Talk.

\bibitem[{Oron and Hoff(2009)}]{OronHoff09}
Oron, A.~P., Hoff, P.~D., 2009. The k-in-a-row up-and-down design, revisited.
  Stat. Med.Accepted for Publication.

\bibitem[{Pisters et~al.(2004)}]{PistersEtAl04}
Pisters, P.~W., et~al., 2004. Phase {I} trial of preoperative doxorubicin-based
  concurrent chemoradiation and surgical resection for localized extremity and
  body wall soft tissue sarcomas. J Clin Oncol 22~(16), 3375--3380.

\bibitem[{Potter(2002)}]{Potter02}
Potter, D.~M., 2002. Adaptive dose finding for phase {I} clinical trials of
  drugs used for chemotherapy of cancer. Stat. Med. 21, 1805--1823.

\bibitem[{{R Development Core Team}(2009)}]{RLang}
{R Development Core Team}, 2009. R: A language and environment for statistical
  computing. R Foundation for Statistical Computing, Vienna, Austria, {ISBN}
  3-900051-07-0.
\newline\urlprefix\url{http://www.R-project.org}

\bibitem[{Robertson et~al.(1988)Robertson, Wright, and
  Dykstra}]{RobertsonEtAl88}
Robertson, T., Wright, F.~T., Dykstra, R.~L., 1988. Order restricted
  statistical inference. Wiley Series in Probability and Mathematical
  Statistics: Probability and Mathematical Statistics. John Wiley \& Sons,
  Chichester.

\bibitem[{Shen and O'Quigley(1996)}]{ShenOQuigley96}
Shen, L.~Z., O'Quigley, J., 1996. {Consistency of continual reassessment method
  under model misspecification}. Biometrika 83~(2), 395--405.

\bibitem[{Shirayev(1996)}]{Shirayev96}
Shirayev, A., 1996. Probability, Second Edition. Springer, New York, translated
  by R.P. Boas.

\bibitem[{Storer(2001)}]{Storer01}
Storer, B., 2001. An evaluation of phase {I} clinical trial designs in the
  continuous dose-response setting. Stat. Med. 20, 2399--2408.

\bibitem[{Yuan and Chappell(2004)}]{YuanChappell04}
Yuan, Z., Chappell, R., 2004. Isotonic designs for phase {I} cancer clinical
  trials with multiple risk groups. Clin. Trials 1~(6), 499--508.

\bibitem[{Zacks et~al.(1998)Zacks, Rogatko, and Babb}]{ZacksEtAl98}
Zacks, S., Rogatko, A., Babb, J., 1998. Optimal bayesian-feasible dose
  escalation for cancer {P}hase {I} trials. Stat. Prob. Lett. 38, 215--220.

\end{thebibliography}

\end{document}